\documentclass[british,12pt, a4paper,twoside,openright]{article}
\usepackage[T1]{fontenc}
\usepackage[utf8]{inputenc}
\usepackage{color}
\usepackage{latexsym}
\usepackage{indentfirst}
\usepackage{graphicx}
\usepackage{amsmath,amssymb,amsfonts,amsthm,amstext,amscd}
\usepackage{mathrsfs}
\usepackage{fancyhdr}
\usepackage{times}
\usepackage{enumerate}
\usepackage{natbib}
\usepackage[hang,flushmargin]{footmisc}

\makeatletter
\newtheorem{defi}{Definition}[section]

\newtheorem{lem}{Lemma}[section]
\newtheorem{rem}{Remark}[section]
\newtheorem{thm}{Theorem}[section]

\newtheorem{coro}{Corollary}[section]

\newtheorem{conjecture}{Conjecture}[section]  %required by text

\fancypagestyle{plain}{%
\fancyhf{}%
\fancyhead[LO]{\textbf{Science \& Philosophy}
}
\fancyhead[RO]{Volume X(Y), 202Z, pp. 7-xx   }
\fancyfoot[C]{}
\fancyfoot[LE,RO]{ \vspace{1cm} doi:10.23756/sp.v10i1.811  }
}

\usepackage{url}
\usepackage{breakurl}
\usepackage{stmaryrd}
\usepackage{tikz}
\usetikzlibrary{shapes,snakes}
\usepackage{pgfplots}
\usetikzlibrary{decorations.pathreplacing,}
\tikzset{radiation/.style={{decorate,decoration={expanding waves,angle=90,segment length=4pt}}},
         antenna/.pic={
        code={\tikzset{scale=5/10}
            \draw[semithick] (0,0) -- (1,4);% left line
            \draw[semithick] (3,0) -- (2,4);% right line
            \draw[semithick] (0,0) arc (180:0:1.5 and -0.5);
            \node[inner sep=4pt] (circ) at (1.5,5.5) {};
            \draw[semithick] (1.5,5.5) circle(8pt);
            \draw[semithick] (1.5,5.5cm-8pt) -- (1.5,4);
            \draw[semithick] (1.5,4) ellipse (0.5 and 0.166);
            \draw[semithick,radiation,decoration={angle=60}] (1.5cm+8pt,5.5) -- +(0:2);
            \draw[semithick,radiation,decoration={angle=60}] (1.5cm-8pt,5.5) -- +(180:2);
  }}
}
\tikzset{reception/.style={{decorate,decoration={triangles,angle=-90,segment length=6pt}}},
         receiver/.pic={
        code={\tikzset{scale=5/10}
            \draw[semithick] (0,0) -- (1,4);% left line
            \draw[semithick] (3,0) -- (2,4);% right line
            \draw[semithick] (0,0) arc (180:0:1.5 and -0.5);
            \node[inner sep=4pt] (circ) at (1.5,5.5) {};
            \draw[semithick] (1.5,5.5) circle(8pt);
            \draw[semithick] (1.5,5.5cm-8pt) -- (1.5,4);
            \draw[semithick] (1.5,4) ellipse (0.5 and 0.166);
            \draw[semithick,reception,decoration={angle=45}] (1.5cm-24pt,5.5) -- +(0:1);
  }}
}
\tikzset{receptionR/.style={{decorate,decoration={triangles,angle=90,segment length=6pt}}},
         receiverR/.pic={
        code={\tikzset{scale=5/10}
            \draw[semithick] (0,0) -- (1,4);% left line
            \draw[semithick] (3,0) -- (2,4);% right line
            \draw[semithick] (0,0) arc (180:0:1.5 and -0.5);
            \node[inner sep=4pt] (circ) at (1.5,5.5) {};
            \draw[semithick] (1.5,5.5) circle(8pt);
            \draw[semithick] (1.5,5.5cm-8pt) -- (1.5,4);
            \draw[semithick] (1.5,4) ellipse (0.5 and 0.166);
            \draw[semithick,receptionR,decoration={angle=45}] (1.5cm+24pt,5.5) -- +(180:1);
  }}
}

\makeatother

\usepackage{babel}

\begin{document}
\title{\textbf{On the symmetries of electrodynamic interactions}}
\author{Hernán Gustavo Solari\footnote{Departamento de Física, FCEN-UBA and IFIBA-CONICET;
Pabellón I, Ciudad Universitaria (1428) - C.A.B.A - Argentina. email:
solari@df.uba.ar Orcid: 0000-0003-4287-1878}
{\ }and
Mario Alberto Natiello\footnote{Centre for Mathematical Sciences, Lund University.
Box 118, S 221 00 LUND, Sweden. email: mario.natiello@math.lth.se
(corresponding author) Orcid: 0000-0002-9481-7454}}

\date{}

\maketitle

\begin{abstract}
\begin{normalsize}
  \noindent
  The development of relational electromagnetism after Gauss
appears to stop around 1870. Maxwell recognised relational electromagnetism
as mathematically equivalent to his own formulae and called for an explanation
of why so different conceptions have such a large part in common.
We reconstruct relational electromagnetism guided by the No Arbitrariness Principle. 
Lorenz’ idea of electromagnetic waves, together with the ``least action
principle'' proposed by Lorentz are enough to derive Maxwell's
equations, the continuity equation and the Lorentz' force. We show that there must 
be two more symmetries in electromagnetism: a descriptive one expressing
source/detector relations, and another relating perceptions of the same source by
detectors moving with different (constant) relative velocities. The Poincaré
group relates perceived fields by different receivers and Lorentz boosts
relate source/detector perceptions. We answer Maxwell's philosophical
question showing how similar theories can be abduced using different
inferred entities. Each form of abduction implies an interpretation and
a facilitation of the theoretical construction.

%AUTHOR doi inserted:
\footnote{\noindent Received on xxxxxx. Accepted on  yyyyyy. Published on zzzzzz.
doi: 10.23756/sp.v10i1.811. ISSN 2282-7757; eISSN 2282-7765. \copyright The Authors.
This paper is published under the CC-BY licence agreement.}
\end{normalsize}
\end{abstract}

%
%
%\begin{figure}[h]
%\centering
%\includegraphics[height=0.7cm, width=1.9cm]{copy.jpg}
%\end{figure}

\newpage

\section{Introduction}

The notion that science, and in particular physics, does not depend
on philosophical or psychological factors is usually manifested by
scientists and the society at large. However, this view confuses what
science should be with how science is actually practised. Following
Peirce we can say that research stops when doubt is appeased and a
(temporary) belief is reached. The condition for the cessation of
doubt might have psychological and philosophical components. During
the late XIX century and the beginning of the XX century an abrupt
change in this condition can be verified \citep{sola22c} finally
leading to new physical understanding \footnote{Meaning the acceptance of a theory by a community}
and a new epistemology \citep{Solari22-phenomenologico}. There is
a relation of precedence: psychological needs (such as the need for
analogies or to incorporate learned habits) determine, in part, physical
theories which in turn determine philosophy. Denying the existence
of the first link we could claim that the Truth in physics forces
upon us the acceptance of some epistemologies and the rejection of
others. In contrast, for a critical philosophy such as Kant's \citep{kant98}
it is philosophy the science that surveils and, if necessary, corrects
all other human activities. Thus, for critical philosophy the sequence
must be: philosophy controls the sciences and the contributions by
psychological needs of scientists have no place and must be eliminated.

The symmetries of electromagnetic interactions played a central role
in the transformation underwent by physics, and with it by science,
during that period. Expectations imported from Mechanics did not fit
observations of electromagnetic phenomena, in particular the propagation
of electromagnetic interactions and light. Two alternatives circulated
around 1850, namely local propagation through some form of physical
medium in space (the ether) against delayed action at a distance.
The second alternative had faded away by the turn of the century,
although it was never proved wrong. The introduction (and subsequent
elimination) of the ether along with a second ingredient: the expectations
posed by society on science reshaped the way physicists approached
Nature. The progress of the industrial revolution expected science
to be the support of technological development, a goal not necessarily
identical to that of exploring Nature in order to understand it. The
utilitarian view of science advanced at the beginning of the second
industrial revolution in the Prussian empire proclaims its success
some 60 years later. With it comes an a--critical epistemology that
denies philosophy the right to examine the foundations of science
\citep{beiser2014after} as it is actually practised: the utilitarian,
capitalist, science.

Is it the same physics resulting from both forms of construction?
For the case of Mechanics most results coincide \citep{sola18b},
while foundational issues regarding the concept of inertial systems
drastically differ \citep{sola22}. 

The ether failed to provide a sound solution to these problems and
Special Relativity was advanced in 1905, being today the accepted
explanatory framework. However, already in 1867 Ludwig Lorenz suggested
an ether-free description of electromagnetism. While the interpretation
of electrodynamics in terms of special relativity must be rejected
as an acceptable theory under a rational construction \citep{Solari22-phenomenologico},
the success obtained by applying this theory to observable problems
and the absence of an alternative (at least) equally successful, consilient
and coherent \citep{whewell1840philosophy,whewell1858Novum} prevented
the criticism of its foundations. 

The combination of motion and coordinate description of electromagnetic
phenomena has several aspects. At least three elements are usually
present: Observer, Source (Emitter, Primary circuit) and Receiver
(Detector, Secondary circuit). However, not all motions are equally
relevant. The No Arbitrariness Principle (NAP)\citep{sola18b} (elaborating
on the idea that no knowledge about nature depends on arbitrary decisions)
suggests that the only motion that actually can influence results
is that between Source and Receiver. Moreover, in a relational description,
there is no other motion involved and the Observer is either absent
or sorted out through a group of symmetry transformations between
equivalent choices.

In this work we illustrate how these setups can be fully handled.
We assemble Electromagnetic theory in terms of classical epistemology;
hopefully achieving a better matching with experiments than current
theories and higher ``consilience'' \citep{thagard1978} (see also
\citep[p. XXXIX, Aphorism XIV]{whewell1840philosophy}). First, we
derive the set of equations of electromagnetism combining Lorenz'
approach with an ether-free version of Lorentz' action integral, unifying
and surpassing ideas that have not been fully investigated so far.
Further, we relate the electromagnetic description for the case where
source and receiver are at relative rest with the corresponding description
in a situation of relative motion, showing also how potentially controversial
concepts such as the ``velocity of light'' ${\displaystyle C=\left(\mu_{0}\epsilon_{0}\right)^{-\frac{1}{2}}}$
in different states of relative motion fit in this nineteenth century
framework. From the concept of reciprocal action (which is in the
philosophical basis of Newton's mechanics) we examine the arbitrariness
that has to be removed in Electromagnetic theory and then, the symmetry
groups that must be involved a-priori. This rational \footnote{The rational epistemology was presented by William Whewell \citep{whewell1840philosophy,whewell1858Novum,whewell1858History}
and further developed by Charles Peirce \citep{peir94} and its fundaments
where available by 1858 before the seminal works of Maxwell \citep{maxw64}
and Lorenz \citep{lore67}.}theory of Electromagnetism does not require any change in space-time
or epistemology.

We try to develop a method that allows all philosophers to grasp its
contents, thus rescuing physics from elitism. If science is to help
us to come into harmony with the universe, beginning with Planet Earth,
a new perspective of exemplary science must be reached, one aiming
at understanding and empathising with all living forms. Thus, the
aim of this work is political, but yet it is philosophical as well
as technical. If successful in our task (as we believe we are), we
can claim that there is no need to abandon the goal of understanding
nature and also that the utilitarian science aimed at ``dominating
nature'' (a prediction technique whose value is given by predictive
success), needs to be left behind if harmony in Planet Earth is our
goal.

\section{On symmetries\label{sec:On-sym}}

Physics sustains the idea that there is a world that reaches us through
the senses and is independent of the observer: the sensed-real. Although
every particular observation may depend on the observer, the collection
of observations points towards a common idea that we call reality,
or \emph{the real}. Thus, the relation between the sensed-real and
reality (the idealisation) plays a fundamental role. This starting
point has been called ``The fundamental antithesis of philosophy''
\citep[Ch. I]{whewell1858History}. Going from the sensed-real to
the real we must separate what belongs to reality from its circumstances
that result in particularities, which quite often are the consequence
of arbitrary decisions. Thus, we reserve the name of arbitrariness
for the observational and descriptive decisions that we have to make
when associating an ideal relation with an observable relation.

It would be desirable to present physical laws in pure abstract form,
without any arbitrary element, but it would be desirable as well,
for physical laws to be as accessible as possible to the mind. Since
abstraction imposes difficulties in grasping the meaning of such laws,
there is a trade-off that must be worked out between the two desires.
This trade-off results in the introduction of some (usually small)
set of arbitrary elements in the description, under the requirement
that such arbitrary elements could be eventually suppressed from the
presentation or, what is the same, that a change in the choice of
arbitrary elements results in an equivalent presentation. These ideas
lead immediately to the existence of a group of transformations relating
different choices of arbitrary elements. The group structure is the
result of the composition law of the transformation between presentations
of the laws under different arbitrary decisions. This is the central
idea under the ``No arbitrariness principle'' (NAP) \citep{sola18b}.

The introduction of an observer brings about the possibility of attaching
to it a Cartesian space for the description of the real and at the
same time it introduces the symmetries of the space (the arbitrary
element).

Moving directly into electromagnetism, we observe that all its fundamental
experiments reflect the influence of electromagnetic phenomena associated
to a pair of bodies (one of them labelled primary circuit, source,
emitter, etc., and the other secondary circuit, receiver, detector).
In the same form that space is not a possible subject of experimental
detection but spatial relations can be measured, electromagnetic fields
can only be detected by their effects on measuring devices, i.e.,
detectors. If the action of a source on a receiver can be addressed
with controlled degrees of influence from the rest of the universe,
in the limit of no influence, the idealised law describing the universe
of such relations must depend only on the relative position and motion
of source and receiver. Such notions can be found all over the foundational
work of Faraday \citep{fara39,fara44,fara55} and Maxwell \citep{maxw73}.

Electromagnetic phenomena imply the motion of electricity (whatever
electricity is, as Maxwell often said) and then, since what changes
the motion of bodies has been called \emph{forces}, we can associate
forces with the action of an electromagnetic (EM) body onto another
EM body. Actually, this use entails a generalisation of the concept
of force, since Newtonian forces change the motional status of macroscopic
bodies while microscopic (quantum) objects, such as electrons involved
in conduction currents, are not what classical mechanics had in mind
when Newton developed its laws. Moreover, if we envisage EM-forces
as Lorentz did, by adopting Weber's view of electrical atoms \citep{lorentz1892CorpsMouvants},
such forces must be identically described by observers whose motions
relate by Galilean coordinate transformations, and furthermore reciprocal
action must be expressed as a symmetry in some privileged systems
we call ``inertial frames'' \citep{thomson1884}. For example, the
symmetry inherent to Newton's third law is expressed as the equation
$F_{12}+F_{21}=0$ being invariant in front of Galilean changes of
coordinates (where $F_{ij}$ is the force on body $j$ originated
in the interaction with body $i$). Yet, we know at least since \citet{poin00}
(see \citep{sola18b} as well) that Newton's ``action and reaction
law'' is not compatible with delayed action at distance. As far as
we know, the form this symmetry takes in EM has not been shown so
far. We will display its effects in the present work.

When EM theory is moved from its original setting as an \emph{interaction
theory} into a \emph{field theory}, some symmetry is broken since
there are no longer two EM-bodies in reciprocal action but we are
thereafter concerned with only one of them, most frequently the \emph{source}.
This presentation of EM may be called the \emph{S-field}. With equivalent
arbitrariness we could shift the focus to the receiver and consider
an \emph{R-field} description. Both descriptions refer to the same
EM phenomena and are therefore related.

When the S-field, the field produced by the source, is perceived by
the source itself or by any extended EM-body not moving with respect
to the source, we call it \emph{S-by-S-field}. When considering the
same S-field as it is perceived by the receiver, we have the \emph{S-by-R-field}
description, see Figure \ref{fig:SbyX} (see Figure \ref{fig:Field descriptions}
for the corresponding \emph{R-field} description). The operation performed
on the description of the phenomenon is to identify one body or the
other with an extended EM-body in the reference frame of the observer.
As both approaches describe the same action, a transformation, possibly
dependent on the relative velocity between the EM-bodies, must relate
their expressions.

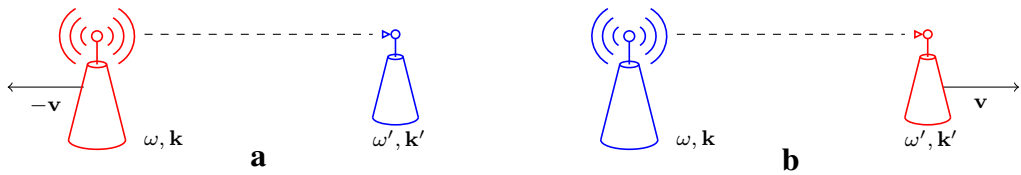
\begin{figure}[h]
\begin{tikzpicture} 
\draw (1.25,0) node {\scriptsize $\omega,\mathbf{k}$};
\draw (4.35,0) node {\scriptsize $\omega^\prime,\mathbf{k}^\prime$};
\draw[<-] (-0.8,0.7)  -- (0.2,0.7);
\draw[below] (-0.3,0.7) node {\scriptsize $-\mathbf{v}$};  
\path (0,0) pic[scale=0.5,color=red] {antenna};   
\path (4,0.3) pic[scale=0.4,color=blue] {receiver}; 
\draw[thin,dashed] (1.,1.4) -- (4,1.4) ; 
\draw (8.25,0) node {\scriptsize $\omega,\mathbf{k}$};
\draw (11.35,0) node {\scriptsize $\omega^\prime,\mathbf{k}^\prime$};
\draw[->] (11.5,0.7)  -- (12.5,0.7);
\draw[below] (12,0.7) node {\scriptsize $\mathbf{v}$};  
\path (7,0) pic[scale=0.5,color=blue] {antenna};   
\path (11,0.3) pic[scale=0.4,color=red] {receiver}; 
\draw[thin,dashed] (8,1.4) -- (11,1.4) ; 
\draw (2.5,-0.25) node {\bf a};
\draw (9.5,-0.25) node {\bf b};
\end{tikzpicture}

\caption{Field of the source (a) as seen by the receiver (S-by-R-field) and
(b) as seen by the source (S-by-S-field). Source to the left of each
image. In blue: the device at rest with the observer.\label{fig:SbyX}}
\end{figure}

For the case of multiple receivers we may want to consider the relation
among the different S-by-R$_{i}$-field descriptions of each receiver.
To connect $R_{1}$ with $R_{2}$ corresponds to the composition of
the transformations between each receiver and the source, namely $R_{1}\to S$
and (the inverse of) $R_{2}\to S$. The composition of transformations
yields a transformation between receivers, that will depend on the
relative velocities of $R_{1}$ and $R_{2}$ with respect to the source.
However, receiver-receiver transformations relate objects of equivalent
character, they are automorphisms and must form a group as well.

\begin{figure}[h]
\begin{tikzpicture} 
\draw (1.25,0) node {\scriptsize $\omega,\mathbf{k}$};
\draw (4.35,0) node {\scriptsize $\omega^\prime,\mathbf{k}^\prime$};
\draw[<-] (-0.8,0.7)  -- (0.2,0.7);
\draw[below] (-0.3,0.7) node {\scriptsize $-\mathbf{v}$};  
\path (0,0) pic[scale=0.5,color=red] {receiverR};   
\path (4,0.3) pic[scale=0.4,color=blue] {antenna}; 
\path[draw,->,decorate,decoration={snake,pre length=3pt, post length=7pt}](0.9,1.4) -- (3.9,1.4) ; 
\draw (8.25,0) node {\scriptsize $\omega,\mathbf{k}$};
\draw (11.35,0) node {\scriptsize $\omega^\prime,\mathbf{k}^\prime$};
\draw[->] (11.5,0.7)  -- (12.5,0.7);
\draw[below] (12,0.7) node {\scriptsize $\mathbf{v}$};  
\path (7,0) pic[scale=0.5,color=blue] {receiverR};   
\path (11,0.3) pic[scale=0.4,color=red] {antenna};
\path[draw,->,decorate,decoration={snake,pre length=3pt, post length=7pt}](7.9,1.4) -- (10.9,1.4) ; 
\draw (2.5,-0.25) node {\bf a};
\draw (9.5,-0.25) node {\bf b};
\end{tikzpicture}

\caption{Field of the receiver (a) as seen by the receiver (R-by-R-field) and
(b) as seen by the source (R-by-S-field). Source to the left of each
image. In blue: the device at rest with the observer.\label{fig:Field descriptions}}
\end{figure}
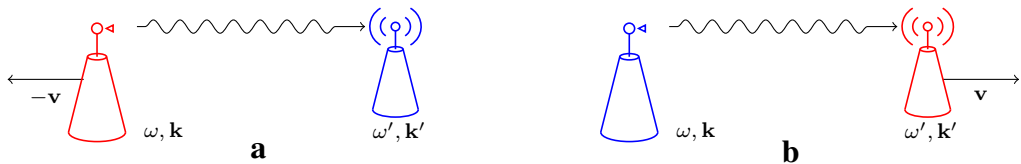

The perceived fact that electromagnetic disturbances require some
time to propagate between source and receiver is acknowledged by all
existing theoretical frameworks of EM. To describe this fact, the
concept of \emph{delayed action at a distance} was advanced in an
organised form by the Danish scientist Ludwig Lorenz \citep{lore67}
after preliminary attempts \citep{bett67,riem67,neum68} from the
Göttingen school originated by ideas of Gauss \citep[bd.5 p. 627-629,][]{gaus70}. 

Returning to relative motion, it must be noticed that even in the
case where source and receiver are in constant  relative motion, the
transformation between the S-by-S-field and S-by-R-field will not
be an inertial transformation (i.e., a Galilean coordinate change).
Galilean transformations correspond to descriptive transformations
that are not concerned with the observable relative motion of the
bodies. The relative motion of source and receiver is a measurable
part of the physics involved and not an arbitrariness (it is there
independently of the observer). Consider the following experiment:
a source is producing a signal sharply peaked around a given frequency,
$\omega_{0}$ as perceived by a receiver not moving with respect to
the source. A set of several, identically built and calibrated receivers
are put in motion at various velocities, $v_{i}$, with respect to
the source, see Figure\ref{fig:Senders-and-receivers}. How is the
signal perceived by each receiver? Which is the perceived characteristic
frequency $\omega_{i}$ ? Which is the relation between the signals
registered by the various receivers?

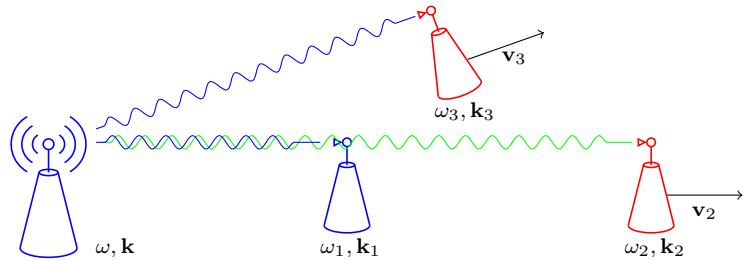
\begin{figure}[h]
\begin{center}
\begin{tikzpicture} 
\draw (1.25,0) node {\scriptsize $\omega,\mathbf{k}$};
\draw (4.35,0) node {\scriptsize $\omega_1,\mathbf{k}_1$};
\draw (8.35,0) node {\scriptsize $\omega_2,\mathbf{k}_2$};
\draw[->] (8.5,0.7)  -- (9.5,0.7);
\draw[below] (9,0.7) node {\scriptsize $\mathbf{v}_2$};
\draw (5.85,1.8) node {\scriptsize $\omega_3,\mathbf{k}_3$};
\draw[->] (5.89,2.5)  -- (6.89,2.855555);
\draw[below] (6.5,2.75) node {\scriptsize $\mathbf{v}_3$};   
\path (0,0) pic[scale=0.5,color=blue] {antenna};   
\path (4,0.3) pic[scale=0.4,color=blue] {receiver};
\path (8,0.3) pic[scale=0.4,color=red] {receiver}; 
\path (5.5,2.01) pic[scale=0.4,color=red,rotate=20] {receiver}; 
\path[draw,green,decorate,decoration={snake,pre length=3pt, post length=7pt}] (1.075,1.4) -> (8.05,1.4) ; 
%\path[draw,blue,decorate,decoration={snake,pre length=3pt, post length=7pt}] (4,1.4) -> (8,1.4) ; 
\path[draw,blue,decorate,decoration={snake,pre length=3pt, post length=7pt}] (1.,1.4) -> (3.95,1.4) ;
\path[draw,blue,decorate,decoration={snake,pre length=3pt, post length=7pt}] (1,1.577778) -> (5.2,3.07111) ; 
\end{tikzpicture}
\end{center}

\caption{Sources and receivers. Blue receiver at rest relative to source, red
receivers in relative motion with respect to the source.\label{fig:Senders-and-receivers}}
\end{figure}

\section{Relational Electrodynamic Background\label{sec:Relational-Electrodynamic-Backgr}}

\subsection{Interaction-based relational formulation.\label{subsec:relational}}

In the presence of electromagnetic interactions, the observable effects
of the interaction can be interpreted as the result of the action
of the \emph{Lorentz force }\citep{lorentz1892CorpsMouvants,natiello2021relational}
over the electrified particles that constitute matter. 

The origins of the Lorentz force can be traced back to Maxwell and
what he called the \emph{Electromotive intensity} \citep[{[598],}][]{maxw73}.
Similarly, Lorentz referred to Maxwell's electrokinetic and potential
energies \citealp[{[630,631] and [634,635],}][]{maxw73}, combining
them in an action integral and the principle of least action. These
presentations take support in Maxwell equations,

\begin{eqnarray}
B & = & \nabla\times A\label{eq:magneticfield}\\
E & = & -\frac{\partial A}{\partial t}-\nabla V\label{eq:electricfield}\\
\epsilon_{0}\nabla\cdot E & = & \rho\label{eq:charges}\\
\mu_{0}j+\frac{1}{C^{2}}\frac{\partial E}{\partial t} & = & \nabla\times B\label{eq:displacement}
\end{eqnarray}
although their derivations some way or the other involved the ether
in the argumentation: Maxwell when considering the ``total current''
of eq.(\ref{eq:displacement}) and Lorentz in the variational principle.

Ludwig Lorenz avoided to introduce the ether by acknowledging that
light was a form of EM interaction and it corresponded with a transversal
wave \citep{lore61,lore1863}, later introducing retarded electromagnetic
potentials \citep{lore67} inspired in Franz Neumann\citep{neum1846-induction}\footnote{Maxwell's results have the same starting point in Neumann's work \citep[{[542],}][]{maxw73}},
\begin{equation}
(A,\frac{V}{C})(x,t)=\frac{\mu_{0}}{4\pi}\int\left(\frac{(j,\rho C)(y,t-\frac{1}{C}|x-y|)}{|x-y|}\right)d^{3}y,\label{eq:Lorenz}
\end{equation}
as an expression based upon these observations, and also on Neumann's
results and Kirchhoff results regarding EM waves in conductors which
make ample use of the continuity equation, ${\displaystyle \frac{\partial\rho}{\partial t}+\nabla\cdot j=0}$.
The displacement equation \eqref{eq:displacement} can be derived
from Equation \eqref{eq:Lorenz} and the continuity equation. It is
everywhere assumed that the current-charge vanishes rapidly enough
at infinity (so that the partial integrations usually present in EM
theory can actually be performed).

In terms of differential equations, Eq. \eqref{eq:magneticfield}
and \eqref{eq:electricfield} are definitions of the magnetic and
electric fields and the main constitutive equation reads
\begin{equation}
\square(A,\frac{1}{C}V)=-\mu_{0}(j,C\rho).\label{eq:diff-wave}
\end{equation}
where ${\displaystyle \square=\Delta-\frac{1}{C^{2}}\frac{\partial^{2}}{\partial t^{2}}}$
is the D'Alembert operator. This equation is satisfied also by:
\begin{equation}
(\widetilde{A},\frac{\widetilde{V}}{C})(x,t)=\frac{\mu_{0}}{4\pi}\int\left(\frac{(j,\rho C)(y,t+\frac{1}{C}|x-y|)}{|x-y|}\right)d^{3}y.\label{eq:Lorenz+}
\end{equation}

The potentials $A,V$ describe the relation between current-density
$j$ or a charge-density $\rho$ with their electromagnetic effect.
The standard interpretation is that $(j,C\rho)$ are the source (the
primary circuit) of the EM action while the potentials are intermediate
fields that indicate their action over the secondary circuit, corresponding
to delayed action; this is, $(A,{\displaystyle \frac{V}{C}})$ are
source fields, S-fields. A different association is possible for $(j,C\rho)$;
they can be interpreted as those corresponding to the secondary circuit
and in such case $(\widetilde{A},{\displaystyle \frac{\widetilde{V}}{C}})$
are the R-fields that sense an EM perturbation away from the receiver
and express its effect later in it, this is, they are advanced fields.

When relevant, we use the indices $1$$(2)$ for the source (receiver).
It is possible to perform a derivation of the Lorentz force \citep{natiello2021relational}
from the Principle of Least Action supported in Maxwell's energy considerations
following Lorentz but using mathematical deduction at the few situations
where Lorentz used arguments corresponding to the ether in \citep{lorentz1892CorpsMouvants}.
Let $\bar{x}(t)$ denote the distance between a reference point in
the source and a reference point in the receiver. We will consider
situations where source and receiver move as rigid bodies in relative
motion (but not in relative rotation) as Lorentz did.

In what follows, $z$ denotes a ``local'' coordinate on body $2$.
We consider, following Lorentz, a collection of virtual displacements
parametrised by time $\delta\bar{x}(t)$ \footnote{As in the Lagrangian formulation, the collection of virtual displacements
is differentiable, i.e., $\dot{\bar{x}}$ exists, and the variation
is zero in the time extremes. Virtual displacements are not the same
as time-dependent perturbations of the position, for the latter have
other effects apart from the change of relative distances. Virtual
displacements are closer to changes of initial conditions than to
perturbations. In particular, during a virtual displacement, there
is no wave progression.}. The variation of charge and current densities \textcolor{blue}{$\rho_{2}(z,t),j_{2}(z,t)$
}on the receiver can be expressed in the coordinates of eq.(\ref{eq:Lorenz})
as:
\begin{eqnarray}
\delta\rho_{2}(x,t) & = & \left(-\delta\bar{x}(t)\cdot\nabla\right)\rho_{2}(x,t)\nonumber \\
\delta j_{2}(x,t) & = & \left(-\delta\bar{x}(t)\cdot\nabla\right)j_{2}(x,t)+\delta\dot{\bar{x}}\rho_{2}(x,t)\label{eq:GalielanVariation}
\end{eqnarray}

The latter relates the local expression of charge and current densities
in the secondary circuit and the same physical object in terms of
the coordinates associated to the primary circuit.

Maxwell considers the electrokinetic and potential energies, which
Lorentz further combines in the action integral
\begin{eqnarray}
\mathcal{A} & = & \frac{1}{2}\int dt\,\int\left(A_{1}(x,t)\cdot j_{2}(x,t)-\rho_{2}(x,t)\,V_{1}(x,t)\right)\,d^{3}x\label{eq:action}
\end{eqnarray}
that here represents the interaction energy between a source or primary
circuit labelled $1$ and a receiver or secondary circuit labelled
$2$. The relation \ref{eq:diff-wave} is satisfied for fields and
current-charge corresponding to the same index. The action integral
in the present form corresponds to an S-by-S-field representation,
namely that the fields of the source are evaluated at the position
of the receiver in the coordinates $x$ of the source and time $t$.

We state the result as a theorem:
\begin{thm}
\label{thm:TheForce} (\citep{natiello2021relational}) Assuming that
all of $|B|^{2}$,$|E|^{2}$, $A$, $j$, $V$, $\rho$ decrease faster
than $\frac{1}{r^{2}}$ at infinity, assuming the action is given
by eq. \ref{eq:action} and given the validity of the continuity equation
${\displaystyle \frac{\partial\rho}{\partial t}+\nabla\cdot j=0}$,
the electromagnetic force

\textup{
\[
F_{em}=\int\!d^{3}x\left[j_{2}(x,t)\times B_{1}(x,t)+\rho_{2}(x,t)\,E_{1}(x,t)\right]
\]
}on the probe can be deduced from Hamilton's principle of minimal
action (\textup{$\delta_{\bar{x}(t)}{\cal A}=0$}) using a virtual
displacement $\delta_{\bar{x}}$ of the probe (which we indicate with
subindex 2), eq.(\ref{eq:GalielanVariation}) with respect to the
primary circuit producing the fields (subindex 1).
\end{thm}

\subsection{Wave equation for the potentials\label{subsec:Wave-equation}}

The wave equation for the potentials can be deduced from Equations
(\ref{eq:magneticfield}-\ref{eq:displacement}).
\begin{lem}
\label{lem:wave}${\displaystyle A(x,t)=\frac{\mu_{0}}{4\pi}\int_{U}\left(\frac{j(y,t-\frac{1}{C}|x-y|)}{|x-y|}\right)\,d^{3}y\Rightarrow\Box A=-\mu_{0}j}$,
and similarly for $\epsilon_{0}\Box V=-\rho$, where $\Box\equiv{\displaystyle \Delta-\frac{1}{C^{2}}\frac{\partial^{2}}{\partial t^{2}}}$.
\end{lem}
For a proof, see Appendix \ref{lem:wave}. Note that this result describes
a property of eq.(\ref{eq:Lorenz}), independently of whether $A,V,j,\rho$
are the electromagnetic vector potential and current, etc., or not.
We prove now that the result holds for the electromagnetic $A,V,j,\rho$,
via a variation of the electromagnetic action \ref{eq:action}:
\begin{thm}
\label{thm:wave}Let $(\mathsf{A},\mathsf{V})$ be the known values
of the electromagnetic potentials in a piece of matter supported on
a region of space with characteristic function $\chi$. Then, assuming
that all of $|B|^{2},|E|^{2},A,j,V,\rho$ decrease faster than $\frac{1}{r^{2}}$
at infinity, Hamilton's principle of least action \citep[Ch 3, 13 A p. 59, ][]{arno89},
$\delta{\cal A}=0$, subject to the constraints given by $(\mathsf{A},\mathsf{V})$
implies the relations
\begin{eqnarray*}
\frac{1}{\mu_{0}}\nabla\times B-\epsilon_{0}\frac{\partial E}{\partial t} & = & \mu_{0}j\\
\epsilon_{0}\nabla\cdot E & = & -\frac{\rho}{\epsilon_{0}}.\\
\nabla\cdot j+\frac{\partial\rho}{\partial t} & = & 0
\end{eqnarray*}
\end{thm}
\begin{coro}
In the special case when the relation $\nabla\cdot A+\frac{1}{C^{2}}\frac{\partial V}{\partial t}=0$
(the ``Lorenz gauge'') is satisfied, the manifestation of the potentials
outside matter obeys the wave equation, eq.(\ref{eq:diff-wave}).
\end{coro}
We develop the proof in Appendix \ref{thm:wave}.

The theorem deserves to be named Lorenz-Lorentz theorem since in Lorenz
conception light was associated to the EM activity inside matter \citep{lore67}
and Lorentz proposed the expression for the action based on Maxwell's
energy considerations.

Recasting the potentials of eq.(\ref{eq:Lorenz}) as the convolution
of charge and currents with the \emph{Lorenz kernel }hereby defined:

\[
K(x-y,s-r)=\frac{1}{|y-x|}\delta(s-r-\frac{1}{C}|y-x|),
\]
namely 
\begin{equation}
(A_{1},\frac{V_{1}}{C})(x,s)=\frac{\mu_{0}}{4\pi}\int\left[\int_{-\infty}^{s}K(x-y,s-r)(j_{1},C\rho_{1})(y,r)\right]\,d^{3}y\,dr\label{eq:kernel}
\end{equation}
a fundamental symmetry between potentials and wave operators is expressed
in the following
\begin{lem}
\label{Lemma:K-wave-inv}The action of the kernel $K(x-y,s-r)$ and
the differential operator $\Box$ are reciprocally inverse of each
other.
\end{lem}
\begin{proof}
We discuss the proof using $A$ to fix ideas, and write eq.(\ref{eq:kernel})
in shorthand as ${\displaystyle A=\frac{\mu_{0}}{4\pi}K\ast j}$ (where
the star stands for convolution). Composition with $\Box$ gives:
\begin{eqnarray*}
\Box A & = & \frac{\mu_{0}}{4\pi}\Box K\ast j=-\mu_{0}j\\
K\ast\Box A & = & -\mu_{0}K\ast j=-4\pi A.
\end{eqnarray*}
Hence, in their respective domain of definition $\Box K=-4\pi Id$
(convolution identity) and $K\ast\Box=-4\pi Id$ (operator identity).
\end{proof}

\subsection{Source/receiver symmetry of the action\label{subsec:Source/receiver-symmetry}}

Since the action \eqref{eq:action} plays a fundamental role in this
relational presentation we should devote some lines to consider its
symmetries.

We first write the action in terms of definite integrals and the kernel
$K(x-y,s-r)$ 
\begin{equation}
\mathcal{A}=\frac{1}{2}\frac{\mu_{0}}{4\pi}\int_{t_{0}}^{t}ds\int_{t_{0}}^{t}dr\,\iint K(x-y,s-r)\left(j_{1}\cdot j_{2}-C^{2}\rho_{1}\rho_{2}\right)d^{3}x\,d^{3}y\label{eq:action-symm}
\end{equation}

The form of the action in eq.(\ref{eq:action-symm}) is almost symmetric
in terms of exchanging primary and secondary circuits. Interchanging
primary and secondary circuit, and $(x,s)\longleftrightarrow(y,r)$
the kernel changes into 

\begin{equation}
K(x-y,s-r)=\frac{1}{|y-x|}\delta(s-r+\frac{1}{C}|y-x|)\label{eq:KerAdv}
\end{equation}
Thus, the action considered is always the action of the primary circuit
over the secondary circuit which can be written in two forms. In one
of them, the S-field (the standard form), EM changes are propagated
with delay by the potentials (and their derivatives, the EM-fields)
at distances away from the source. The symmetry-related form, the
R-field, associates an advanced field with the receiver. In this form,
the field can be seen as a sensor that will carry disturbances to
the receiver that will display changes at a later time.

The symmetry of the action has the immediate consequence that all
lemmas and theorems of subsections \eqref{subsec:relational} and
\eqref{subsec:Wave-equation} have an equivalent form under this symmetry
operation. In particular, there is Lorentz-force where the S-fields,
R-currents and R-charges are exchanged by R-fields, S-currents and
S-charges. This relation is what corresponds to the action and reaction
law for actions that propagate instantaneously, since in the limit
$C\rightarrow\infty$ the S-field and the R-field of a given body/device
coincide.

\subsection{Detection/perception in relative motion \label{subsec:Detection}}

Let us consider the potentials $A,V$ originated in a source with
current-charge $J=(j,C\rho)$ measured at (rest relative to) the source
(with coordinate $y$). We consider further a detector extending over
a variable $x$ with reference to a distinguished point in it. In
the case of source and detector at relative rest, we write 
\begin{eqnarray}
(A,\frac{V}{C})(x,t) & = & \frac{\mu_{0}}{4\pi}\int d^{3}y\,\int ds\left(\frac{\delta((t-s)-\frac{1}{C}|x-y|)}{|x-y|}\right)\,J(y,s)\label{eq:rest}\\
 & = & \frac{\mu_{0}}{4\pi}\int d^{3}z\,\left(\frac{J(x-z,t-\frac{1}{C}|z|)}{|z|}\right)
\end{eqnarray}
These equations are formulated under the following premises: Coordinates
$y$ and $x$ are described from the same spatial reference system
$S$, whatever it is, and hence at a given time $t$, $x-y$ and in
particular $|x-y|$ are objective invariant quantities. Moreover,
since source and detector are in relative rest, these quantities are
independent of $t$. In the present conception of electromagnetism
there is another objective invariant quantity of relevance, namely
the electromagnetic delay ${\displaystyle \Delta_{0}=t-s}$. The index
$0$ highlights the situation of relative rest between source and
detector. It is the state of point $y$ on the source at the previous
time $s$, where $C(t-s)=|x-y|$ what connects with point $x$ of
the detector at time $t$. Finally, the second row displays the change
of variables $z=x-y$.

In order to address detection in relative motion we advance the following 
\begin{conjecture}
A detector recording solely electromagnetic information (e.g. an electromagnetic
wave) cannot determine its relative velocity with respect to the source
(assumed constant).
\end{conjecture}
Consequently, let us postulate that a detector in relative motion
with velocity $v$ with respect to the source perceives an EM wave
which cannot be distinguished from the one originating in some current-charge
\emph{at} \emph{relative rest}. We would like to show something like:
\begin{eqnarray}
(A,\frac{V}{C})_{v}(x,t) & = & \frac{\mu_{0}}{4\pi}\int d^{3}y\,\left(\frac{1}{|x-y|}\right)\,J_{v}(y,t-\Delta)\label{eq:motion}
\end{eqnarray}
with $\Delta=\frac{1}{C}|x-y|$.

In this new situation we still have one reference frame $S$ to describe
both source and detector. Again, $z=x-y$ is an objective quantity,
only that now two differences arise: (a) $x-y$ depends on $t$ because
of the relative motion and (b) the electromagnetic delay may be modified
in order to take into account the relative motion. Throughout this
discussion, $t$ is the (present) time when the electromagnetic interaction
is detected, $(x-y)$ indicates the relative position of (points of)
detector and source at time $t$, $\Delta_{v}=(t-s)_{v}$ is the electromagnetic
delay and $(x-y)_{v}$ is the corresponding relative position at time
$s$ when the electrical disturbance in the source took place, and
the index $v\in\mathbb{R}^{3}$ indicates a situation of relative
motion between source and detector. The index $v$ will be some function
of the relative velocity $u$ between source and detector to be determined
in what follows. Moreover, $(x-y)_{v}$ and $\Delta_{v}$ are objective
and invariant quantities, independent of the choice of reference frame.

We intend to find the correspondence between disturbances in the primary
circuit and actions on the secondary system. We begin by considering
an infinitesimal velocity $\delta v$, with $\frac{d\delta v}{dt}=0$.
In this case we have
\begin{defi}
\textbf{(Differential delayed interaction condition) }In the presence
of relative motion with infinitesimal velocity $\delta v$, a disturbance
originated at point $y$ and time $t-\Delta_{\delta v}$ produces
an electromagnetic action at $(x,t)$, where
\begin{eqnarray*}
C\Delta_{\delta v} & = & |x-y-\Delta_{\delta v}\delta v|.
\end{eqnarray*}
\end{defi}
For $\delta v=0$ the condition reduces to $C\Delta_{0}=|x-y|$, corresponding
to Lorenz' potentials, eqs.(\ref{eq:Lorenz}) and (\ref{eq:kernel})\footnote{Letting $s=t-\Delta_{\delta V}$we may read the definition as a consequence
of: $(x-y)(s)=(x-y)(t)-(t-s)\delta v$.}. Note that $C$ enters in both expressions since we postulate that
the detector in relative motion registers an electromagnetic signal
\emph{as if the source were at relative rest}. This definition leads
to the following
\begin{lem}
\label{Lemma:delayed-interaction}Let $(x-y)_{v}$ be the separation
of source and detector at time $s$ when the electrical disturbance
at the source took place in a situation of relative motion labelled
by $v\in\mathbb{R}^{3}$ and $\Delta_{v}$ the corresponding electromagnetic
delay, while $(x-y)_{0},\Delta_{0}$ are the corresponding quantities
for source and detector at relative rest. Then, for each $v$ the
\textbf{delayed interaction condition} satisfies 
\[
\left(\begin{array}{c}
\left(x-y\right)\\
C\Delta
\end{array}\right)_{v}=\exp\left(-\left(\begin{array}{cc}
\mathbf{0} & {\displaystyle \frac{v}{C}}\\
{\displaystyle \frac{v}{C}}^{T} & 0
\end{array}\right)\right)\left(\begin{array}{c}
\left(x-y\right)\\
C\Delta
\end{array}\right)_{0}
\]
\end{lem}
\begin{proof}
To lowest order in $\delta v$ the difference in $\Delta$'s is:
\begin{eqnarray*}
C\left(\Delta_{\delta v}-\Delta_{0}\right) & = & \sqrt{|x-y|^{2}-2\left(x-y\right)\cdot\delta v\Delta_{\delta v}+|\delta v|^{2}\Delta_{\delta v}^{2}}-|x-y|\\
 & = & -\frac{\left(x-y\right)}{|x-y|}\cdot\delta v\Delta_{0}+O(\delta v^{2})=-\left(x-y\right)\cdot\frac{\delta v}{C}+O(\delta v^{2})
\end{eqnarray*}
In this limiting case the condition reads 
\begin{eqnarray}
\left(\begin{array}{c}
\left(x-y\right)_{\delta v}\\
C\Delta_{\delta v}
\end{array}\right) & = & \left(\begin{array}{c}
(x-y)-\delta v\Delta_{0}\\
C\Delta_{0}-\left(x-y\right)\cdot\frac{\delta v}{C}
\end{array}\right)\label{eq:TL-coord}\\
 & = & \left[\left(\begin{array}{cc}
\mathbf{1} & \mathbf{0}\\
\mathbf{0} & 1
\end{array}\right)-\left(\begin{array}{cc}
0 & {\displaystyle \frac{\delta v}{C}}\\
({\displaystyle \frac{\delta v}{C}})^{T} & 0
\end{array}\right)\right]\left(\begin{array}{c}
(x-y)\\
C\Delta_{0}
\end{array}\right).\nonumber 
\end{eqnarray}
In other words, there exists an infinitesimal transformation on $\mathbb{R}^{3+1}$
connecting the condition for $v=0$ with that for $\delta v$. By
the Trotter product formula we obtain Lie's result for finite $v$
as a repeated composition of infinitesimal shifts,
\begin{eqnarray}
TL(-v) & \equiv & \exp\left(-\left(\begin{array}{cc}
\mathbf{0} & {\displaystyle \frac{v}{C}}\\
{\displaystyle \frac{v}{C}}^{T} & 0
\end{array}\right)\right)\nonumber \\
 & = & \lim_{n\rightarrow\infty}\left[\left(\begin{array}{cc}
\mathbf{1} & \mathbf{0}\\
\mathbf{0} & 1
\end{array}\right)-\frac{1}{n}\left(\begin{array}{cc}
\mathbf{0} & {\displaystyle \frac{v}{C}}\\
{\displaystyle \frac{v}{C}^{T}} & 0
\end{array}\right)\right]^{n}\label{eq:groupalgebra}
\end{eqnarray}
thus proving the statement.
\end{proof}
\begin{rem}
Explicit formulae for the Lorentz transformations are shown in the
Appendix \ref{sec apendice LT}. The more familiar form $L(u)$ of
the transformation is displayed in 
\begin{eqnarray}
\left(\begin{array}{c}
z_{u}\\
C\Delta_{u}
\end{array}\right) & = & L(u)\left(\begin{array}{c}
z\\
C\Delta_{0}
\end{array}\right)=\left(\begin{array}{c}
z+(\gamma-1)\hat{u}\left(\hat{u}\cdot z\right)+\gamma{\displaystyle \frac{u}{C}}C\Delta_{0}\\
\gamma\left(C\Delta_{0}+{\displaystyle \frac{u\cdot z}{C}}\right)
\end{array}\right),\label{eq:Lform}
\end{eqnarray}
where $u={\displaystyle C\hat{v}\tanh\left|\frac{v}{C}\right|}$ and
we use the shorthand $x-y=z$. There is a 1-to-1 correspondence in
Lemma \ref{Lemma:delayed-interaction}, between the two presentations
of the Lorentz transformations, namely $TL(-v)\equiv L(-u)$. Hence,
we will use only $u$ in the sequel. $u$ is interpreted as the relative
velocity between source and detector.\\
The basis for the interpretation of $u$ as the relative velocity
is as follows. Consider the vector space $\mathbb{R}^{3+1}\equiv\mathbb{R}^{3}\times\mathbb{R}$
associated to relative positions and relative time. A Lorentz transformation
(LT), eq.\eqref{eq:Lform}, as well as a Galilean transformation GT,
\[
\left(\begin{array}{c}
Z^{\prime}\\
T^{\prime}
\end{array}\right)=\left(\begin{array}{cc}
\mathbf{1} & u\\
0 & 1
\end{array}\right)\left(\begin{array}{c}
Z\\
T
\end{array}\right)
\]
\end{rem}
can be regarded as endomorphisms of $\mathbb{R}^{3+1}$ mapping a
situation at relative rest onto a situation of relative motion. While
the velocity $u$ in the GT has a mechanical origin, the parameter
$u$ in LT is an abstract parameter used to classify transformations
and a point of contact with the underlying physical problem is required
to furnish a physical interpretation to the LT's. Considering lines
on $\mathbb{R}^{3+1}$ associated to a fixed relative position, $Z$
and different time-intervals, we obtain for the Galilean transformation
the (physical) relative velocity $u=\frac{Z^{\prime}(T_{1})-Z^{\prime}(T_{0})}{T^{\prime}(T_{1})-T^{\prime}(T_{0})}$
while in the case of the Lorentz transformation we obtain
\begin{eqnarray*}
\frac{z^{\prime}(\tau_{1})-z^{\prime}(\tau_{0})}{\tau^{\prime}(\tau_{1})-\tau^{\prime}(\tau_{0})} & = & \frac{\gamma u(\tau_{1}-\tau_{0})}{\gamma(\tau_{1}-\tau_{0})}=u.
\end{eqnarray*}
While the GT preserves times and as such can be viewed as a transformation
in relational-space only, the LT preserves $|z|^{2}-(C\tau)^{2}$
and, as a particular case, the condition of being in electromagnetic
contact, $|z|^{2}-(C\tau)^{2}=0$. We may associate the same relational
velocity to both GT and LT.

Eq.(\ref{eq:groupalgebra}) displays the action of a Lorentz' boost
\citep{gilm74} in the Lie algebra (rhs) and group (lhs). The generators
of the Lorentz boosts plus the generators of the rotations constitute
the basis of the Lie algebra which exponentiated gives the Poincaré-Lorentz
group. While the spatial rotations form a subgroup of the Poincaré-Lorentz
group, the Lorentz boosts do not. Any element of the Poincaré-Lorentz
group can be written as a product: $P=L(u)R(\Omega)$ as well as $P=R(\Omega)L(u^{\prime})$
being $\Omega$ a 3d-rotation and $u^{\prime}=R(\Omega)u$ . These
forms are known as left and right coset decompositions of the group
\citep{hame62,gilm74}.
\begin{rem}
By construction of the $LT's,$there is an upper limit for having
electromagnetic contact amenable to be related with situations at
relative rest. While there is no mechanical limit to relative velocity,
the present theory describes electromagnetic interactions only for
$|u|<C$ .
\end{rem}
\begin{rem}
Eqs.\ref{eq:TL-coord} and \ref{eq:Lform} for the detector and source
points, $x,y$ which are in electromagnetic interaction at time $t$,
display their relative position $(x-y)_{u}$ at the time $t-\Delta_{u}$
when the disturbance in the source took place. The ratio $\frac{|(x-y)_{u}|}{\Delta_{u}}=C$
is always satisfied by construction.
\end{rem}
Next, we note that the propagation kernel can be more properly written
as
\[
K=\left\{ \begin{array}{cc}
0, & \quad\left(t-s\right)<0\\
{\displaystyle \frac{\delta(t-s-\frac{1}{C}|x-y|)}{|x-y|}}, & \quad\left(t-s\right)\ge0.
\end{array}\right.
\]
Hence, we have the following
\begin{lem}
\textbf{(Symmetric form of the propagation kernel) }\label{Lem:sym}
Lorenz propagation kernel can be rewritten as

\begin{equation}
K=\left\{ \begin{array}{cc}
0, & \quad\left(t-s\right)<0\\
{\displaystyle \frac{2}{C}\delta(\left(t-s\right)^{2}-\frac{1}{C^{2}}|x-y|^{2})}, & \quad\left(t-s\right)\ge0.
\end{array}\right.\label{eq:kernel-beauty}
\end{equation}
\end{lem}
\begin{proof}
In the distribution sense $K={\displaystyle \frac{2|x-y|}{C(t-s)+|x-y|}}K$.
By another distributional property, for any $g(s)$ such that $g(s_{0})\ne0$
it holds that ${\displaystyle \frac{\delta(s-s_{0})}{|g(s)|}}=\delta(g(s)(s-s_{0}))$.
In this case, $g(s)=t-s+{\displaystyle \frac{1}{C}}|x-y|$. Hence,
we obtain the symmetric kernel expression of eq.\eqref{eq:kernel-beauty}.
\end{proof}
\begin{thm}
\label{Thm-comm} The Lorenz propagation kernel $K(x,t;y,s)$ has
the following properties in relation to Lorentz transformations
\begin{eqnarray*}
K(L_{u}(x,t);L_{u}(y,s)) & = & K(x,t;y,s)\\
K(L_{u}(x,t);y,s) & = & K(x,t;L_{-u}(y,s))\\
\int d^{3}y\,ds\left[K(L_{u}(x,t);y,s)J(y,s)\right] & = & \int d^{3}y\,ds\left[K(x,t;y,s)J(L_{u}(y,s))\right]
\end{eqnarray*}
The last equation reads: the transformation of the potentials are
the potentials associated to the transformations of the currents.
We say then that the linear operator associated with $K$ commutes
with the Lorentz transformation.
\end{thm}
\begin{proof}
It is straightforward to verify that the argument of the $\delta$-distribution
in eq.\eqref{eq:kernel-beauty} is invariant upon Lorentz transformations,
namely that if\\
 $\left((x-y)_{u},C(t-s)_{u}\right)$ satisfy eq.(\ref{eq:Lform}),
then $\left(t-s\right)^{2}-{\displaystyle \frac{1}{C^{2}}}|x-y|^{2}=\left(t-s\right)_{u}^{2}-{\displaystyle \frac{1}{C^{2}}}|(x-y)_{u}|^{2}$
and also $(t-s)\ge0\Longleftrightarrow(t-s)_{u}\ge0$. Thus,
\[
K=\left\{ \begin{array}{cc}
0, & \quad\left(t-s\right)_{u}<0\\
{\displaystyle \frac{2}{C}\delta(\left(t-s\right)_{u}^{2}-\frac{1}{C^{2}}|\left(x-y\right)_{u}|^{2})}, & \quad\left(t-s\right)_{u}\ge0.
\end{array}\right.
\]
is independent of $u$. Using the first property it follows that $K(L_{u}(x,t);y,s)=K(L_{u}(x,t;L_{u}L_{-u}(y,s))=K(x,t;L_{-u}(y,s))$.
The commutation relation is the result of integrating the kernel to
produce a linear operator and changing integration variables $((y,s)\mapsto L_{u}(y^{\prime},s^{\prime})$.
\end{proof}
\begin{rem}
The points that are in electromagnetic connection are characterised
by $\left(C(t-s)\right)^{2}-\left|x-y\right|^{2}=0$. Calling $\tau_{u}\equiv(t-s_{0})_{u}$
and $\chi_{u}\equiv(x-y)_{u}$, the interaction kernel is the convolution
kernel of $\delta(\tau_{u}^{2}-(\chi_{u}/C)\textasciicircum2)$ which
can be split in two contributions, one for $\tau_{u}\ge0$ and another
for $\tau_{u}\le0$. But, if $(0,0)$ is influencing $(\tau_{0},\chi_{0})$
for $\tau_{0}\ge0$, it results that $\tau_{u}>0$ (using that $|u\cdot x/C^{2}|=\frac{|u\cdot\chi|}{|\chi||u|}\frac{|u||\chi|}{C^{2}}<\frac{|u||\chi|}{C^{2}}$)
hence the splitting is really in terms of \emph{influencing}, $\tau_{u}\ge0$,
vs. \emph{being influenced,} $\tau_{u}\le0$. This separates the sets
in a form invariant with respect to $u$.
\end{rem}

\subsubsection{Perceived fields and inferred currents-charges}

Examining eq.\eqref{eq:motion}, we note that it represents a convolution
product with convolution kernel $\kappa(z,r)$, with $K(x,t;y,s)=\kappa(x-y,t-s)$
and that
\[
(A,\frac{V}{C})_{u}=\frac{\mu_{0}}{4\pi}\kappa*J_{u}=\frac{\mu_{0}}{4\pi}J_{u}*\kappa
\]
where the convolution is in time and space.

According to eq.\eqref{eq:TL-coord}, the arguments in the current
are $(x-y,t-s)$, for $u=0$. For $u\ne0$ the points that are in
electromagnetic relation according to Lemma \ref{Lemma:delayed-interaction}
are $((x-y)_{u},(t-s)_{u})$, thus in $J_{u}*\kappa$, we propose
\begin{conjecture}
The arguments of the effective current are\textbf{ }$((x-y)_{u},(t-s)_{u})$,
i.e.,
	$J_{u}=L(-u)J(L(u)(x-y,t-s))$, where $J$ is the current-charge
measured by the source.
\end{conjecture}
At this point we must notice that there are three forms in which current-charge
can be transformed to produce a new pair satisfying the continuity
equation. Two of them are Galilean:
\begin{eqnarray}
(j,C\rho)(x,t) & = & (j-v\rho,C\rho)(x+vt,t)\label{eq:GaliA}\\
(j,C\rho)(x,t) & = & ({\displaystyle j,C\rho-\frac{v}{C}\cdot j})(x,t+\frac{v\cdot x}{C^{2}})\label{eq:GaliB}\\
(j,C\rho)_{u}(x,t) & = & L(-u)(j,C\rho)(L(u)(x,t))\label{eq:properJ}
\end{eqnarray}
 In the third form, the leftmost $L$ acts on the charge-current $4D$-vector
while the rightmost acts on the space-time coordinates.

If the form \eqref{eq:GaliA} is adopted, a theorem due to Maxwell
\citep[{[602]}][]{maxw73} shows that from the point of view of the
receiver the transformation \eqref{eq:GaliB} must be applied to preserve
the mechanical force but in such case the perceived potentials/fields
are not waves. The empirical evidence has judged this view as not
correct.

We propose to adopt eq.\eqref{eq:properJ} as a definition of the
inferred current. We insist at this point that the symmetry is not
an a-posteriori observation of the formulae, but rather an a-priori
demand of constructive reason as explained in \citep{sola18b}. The
transformation of current-charge presents itself as a demand of reason
to be later confronted with empirical results. That a charge density
in motion can be perceived as a current is a belief firmly adopted
since Weber's electrodynamic studies \citep{webe46} and we are habituated
to accept it, while that a neutral current in motion will be perceived
as charge is not rooted in our beliefs in the same way, despite the
fact that Maxwell's theorem already opened for that possibility.
\begin{rem}
\label{Rem:SR-fields} The symmetric form of $K$ is especially appealing
when consider the backwards propagation kernel, as in the equation
pairs (\ref{eq:Lorenz})--(\ref{eq:Lorenz+}) and (\ref{eq:kernel})--(\ref{eq:KerAdv}).
The backward propagation kernel is the result of inverting the time
inequalities in \ref{eq:kernel-beauty}.
\end{rem}
\begin{rem}
Which is the meaning of a successive application of Lorentz' transformations
to a current? The meaning we find apt is that if $J_{u}=L_{-u}J(L_{u}(x,t))$,
then $J=L_{u}J_{u}(L_{-u}(x,t))$ (since Lorentz transformations have
as inverse the transformation based on minus the velocity) and correspondingly \\ 
$J_{u^{\prime}}=L_{-u^{\prime}}L_{u}J_{u}(L_{u^{\prime}}L_{-u}(x,t))$.
Since $L_{u^{\prime}}L_{-u}$ is a general element of the Poincaré-Lorentz
group, $L_{u^{\prime}}L_{-u}=L_{u^{\prime}\ominus u}R(u^{\prime},u)$
with $u^{\prime}\ominus u$ the coset addition of velocities, also
known as Einstein's addition \citep{gilm74} and $R(u^{\prime},u)$
a Wigner rotation\footnote{Wigner was not the first to study the group structure associated to
Lorentz transformations or to mention the rotation. At least Silberstein
\citep[p. 167]{silberstein1914} in the published notes of his 1912-1913
course on Relativity at the University College, London, preceded Wigner,
who acknowledged this precedence.} . Thus, the Poincaré-Lorentz group allows to convert between inferred
currents or fields associated to different detectors in relative motion
with respect to the same source. Notice that the relative velocity
between both receptors is $u^{\prime}-u$ but the correspondence of
electromagnetic perceptions is not $L(u^{\prime}-u)$ which might
even not exist.
\end{rem}
Next, we explore the consequences of this proposal.\foreignlanguage{english}{
Let us define operators acting on scalar or vector functions, $J$,
of $(x,t)$ as 
\begin{eqnarray}
\widehat{K}\left[J\right](x,t) & \equiv & \iint d^{3}z\,d\Delta\,K(z,\Delta)J(x-z,t-\Delta)\label{eq:conv}\\
\widehat{L}_{u}\left[J\right] & \equiv & J(L(u)(x,t))\nonumber \\
(\widehat{A}\circ\widehat{B})[J] & \equiv & \widehat{A}\left[\widehat{B}[J]\right]\nonumber 
\end{eqnarray}
The first line defines the action of the propagating kernel as a convolution,
the second the action of a Lorentz transformation on the coordinates
(recall that }$u={\displaystyle Cv\tanh\left|\frac{v}{C}\right|}$)
while the third relation establishes notation.
\begin{lem}
\label{Lemma:Op} According to the previous discussion, the perceived
potentials read
\begin{equation}
(A,\frac{V}{C})_{u}=\widehat{K}\left[J_{u}\right]\label{eq:K-operator-form}
\end{equation}
In addition, we have the following identities 
\begin{eqnarray*}
\widehat{K}\left[J_{u}\right] & = & L(-u)\widehat{L}_{u}\left[\widehat{K}\left[J\right]\right]
\end{eqnarray*}
\end{lem}
\begin{proof}
Note that $L(-u)$ acts on the current-charge $J=(j,C\rho)$, while
$\widehat{L}_{u}$ acts on the spatial/temporal arguments $x,Ct$.
Eq.(\ref{eq:K-operator-form}) is just eq.(\ref{eq:motion}) rewritten
through eq.(\ref{eq:conv}). Recalling from eq.(\ref{eq:properJ})
that $J_{u}=L(-u)\widehat{L}_{u}\left[J\right]$ and from \ref{Thm-comm}
that $\widehat{L}_{u}\circ\widehat{K}=\widehat{K}\circ\widehat{L}_{u}$
and finally that the matrix $L(-u)$ commutes with the scalar operator
$K$ we obtain the result.
\end{proof}

\subsubsection{The Doppler effect}

The perception of wave frequencies in the case the waves are produced
by a source in relative motion with respect to the receptor is known
as \emph{Doppler effect}. The EM Doppler effect plays a fundamental
role in physics \citep{ding60a,mand62,kaiv85}. The goal of this section
is to show that the present theory provides an explanation for the
experimental observations of the Doppler effect. To begin with, all
Doppler experiments consist in comparing the waves perceived by a
detector at rest with respect to the source against the perception
of a detector moving at constant velocity (within acceptable experimental
precision) relatively to the source.

In practice, the task is to obtain the Fourier transform of eq.(\ref{eq:motion}).
We will keep track of this process conceptually, and hence it is better
to use the operator notation from Lemma \ref{Lemma:Op}. The Fourier
transform of a function will be:
\[
{\cal F}_{k,w}\left[\phi\right]=\frac{1}{(2\pi)^{2}}\iint d^{3}x\,dt\,\exp\left(-i(k\cdot x-wt)\right)\phi(x,t)
\]
and is a function of $(k,w)$, where we have made an arbitrary choice
in the election of the sign preceding $wt$ (that does not influence
the conclusion). We will use the following known results:
\begin{eqnarray*}
{\cal F}_{k,w}\left[\widehat{L}_{u}\phi\right] & = & {\cal F}_{k^{\prime},w^{\prime}}\left[\phi\right]\;,\mbox{with }(k^{\prime},{\displaystyle \frac{w^{\prime}}{C}})=L(-u)(k,{\displaystyle \frac{w}{C}})\\
{\cal F}_{k,w}\left[\widehat{K}\phi\right] & = & \frac{1}{w^{2}-C^{2}k^{2}}{\cal F}_{k,w}\left[\phi\right]
\end{eqnarray*}
The first result is the immediate consequence of $L(u)$ being symmetric,
while the second one can be obtained in various ways including direct
integration. Applying these results to eq.(\eqref{eq:K-operator-form})
we obtain

\begin{eqnarray*}
{\cal F}_{k,w}\left[\widehat{K}[J_{u}]\right] & = & {\cal F}_{k,w}\left[\widehat{K}\left[L(-u)\widehat{L}_{u}\left[J\right]\right]\right]\\
 & = & L(-u){\cal F}_{k,w}\left[\widehat{L}_{u}\left[\widehat{K}\left[J\right]\right]\right]\\
 & = & L(-u){\cal F}_{k^{\prime},w^{\prime}}\left[\widehat{K}\left[J\right]\right]\\
 & = & L(-u)\frac{1}{w^{\prime2}-C^{2}k^{\prime2}}{\cal F}_{k^{\prime},w^{\prime}}\left[J\right]\\
 & = & L(-u)\frac{1}{w^{2}-C^{2}k^{2}}{\cal F}_{k^{\prime},w^{\prime}}\left[J\right]
\end{eqnarray*}
where $(k^{\prime},{\displaystyle \frac{w^{\prime}}{C}})=L(-u)(k,{\displaystyle \frac{w}{C}})$.
Thus, in terms of wave frequencies, the Fourier spectrum will have
a peak at $w^{\prime}=\gamma(u)(w-k\cdot u)$ associated with a source
of frequency $w$. The primed quantities describe the characteristics
of the wave as perceived by the detector while the unprimed refer
to the source. When $k\cdot u=|k||u|$ the relative distance between
source and detector increases, $w^{\prime}<w$, and correspondingly
the wavelength shifts towards higher values (red shift). 

Hence, we have proved the following
\begin{thm}
\textbf{(Doppler effect}\label{Doppler-effect}\textbf{)} A detector
(observer) in relative motion with velocity $u$ with respect to an
electromagnetic source emitting current-charge waves of wavelength
and frequency $(k,w)$ detects electromagnetic waves of wavelength
and frequency ${\displaystyle (k^{\prime},\frac{w^{\prime}}{C})=}L(-u)(k,{\displaystyle \frac{w}{C})}$.
\end{thm}
\begin{rem}
The symmetry \eqref{eq:LorGen} corresponds to expressing the action
in terms of the inferred charge and currents by an observer. As such,
it corresponds to a subjective view of EM.
\end{rem}
\begin{rem}
The Galilean variation that allowed us to obtain the Lorentz force
from the action, eq.\eqref{eq:GalielanVariation}, indicates that
the force experienced by the moving circuit takes the same form but
the potentials to be used correspond to the perceived potentials of
eq.\eqref{eq:motion}.
\end{rem}

\subsection{Mathematical presentation of the Lorentz transformation as a symmetry\label{subsec:Lorentz-extended}}

Since Lorentz' transformations are well known in relation to electromagnetism,
we consider their effect on the action and find their meaning in the
present context.

Let ${\cal I}$ be the infinitesimal generator for the Lorentz transformation
\begin{equation}
{\cal I}_{j}=\left(Ct\frac{\partial}{\partial x_{j}}+\frac{x_{j}}{C}\frac{\partial}{\partial t}\right)\label{eq:LorGen}
\end{equation}
which together with the generators of the rotations
\[
{\cal J}_{i}=\sum_{jk}\epsilon_{jki}\left(x_{k}\frac{\partial}{\partial x_{i}}-x_{i}\frac{\partial}{\partial x_{k}}\right)
\]
(with $\epsilon_{jki}$ Kronecker's antisymmetric tensor) complete
the Lie algebra of the Poincaré-Lorentz group \citep{gilm74}.
\begin{thm}
\label{Thm-NoLor} The electromagnetic action ${\cal A}$ \eqref{eq:action}
transforms into an equivalent action ${\cal A}^{\prime}$ when the
infinitesimal transformations 
\[
\hat{\delta}=\sum_{i}\left(\delta\theta_{i}{\cal J}_{i}+\delta v_{i}{\cal I}_{i}\right)
\]
operate on $\left(j,C\rho\right)_{1,2}$ simultaneously and $\frac{d\delta v_{i}}{dt}=0$.
\end{thm}
\begin{proof}
The result follows from the observation that the kernel $K$ in \eqref{eq:kernel}
commutes with the six generators as a result of Theorem \eqref{Thm-comm},
and that, integrating by parts in space and time the action of $\hat{\delta}$
over $\left(j,C\rho\right)_{2}$ can be seen as an action over $\left(j,C\rho\right)_{1}$
preceded by a negative sign, and then, both actions compensate to
first order. Thus, the infinitesimal action of any element of the
Lie algebra acting on both subsystems (primary and secondary) corresponds
to the identity. We have that
\[
{\cal A}={\cal A}^{\prime}+{\cal F}(t)
\]
 with ${\cal F}(t)$ a functional of the potentials and currents evaluated
at the time $t$. Since all variations are considered to be zero at
the extremes of the time-interval, ${\cal F}(t)$ contributes to zero
to the variational calculation. In terms of their variations, ${\cal A}$
and ${\cal A}^{\prime}$ are equivalent. See the Appendix \eqref{subsec:Proof-of-NoLor}
for the algebraic details.
\end{proof}
The requirement for $\delta v$ to be constant in time is familiar
to any one acquainted with Lorentz' transformations. It is interesting
to mention that in the present context this requirement can be lifted
by defining the variation as
\begin{equation}
\hat{\delta}_{a}=\sum_{i}\left(\delta\theta_{i}{\cal J}_{i}+\delta v_{i}{\cal I}_{i}+\frac{1}{2C}\frac{d\delta v_{i}}{dt}x_{i}\right)\label{eq:extended-Lorentz}
\end{equation}

\section{Discussion and conclusions}

From the point of view of pragmaticist epistemology \citep{peir55}
all currents and charges are inferred. What we know about them are
their effects, hence charge and current are ``that what produces
this and such effects'', i.e., ideas, inferred entities, not directly
accessible to our senses. However, charges and currents were originally
associated to forces measured by a torsion balance and deflections
of needles observed in galvanometers. Such primitive methods constitute
the original definition of currents and charges and are available
only for an observer at rest with the measuring apparatus since they
are based upon material connections of circuits. In the text, we have
restricted the use of ``inferred'' to those measurements that are
performed based on action at a distance, i.e., without a ``material''
connection between the circuits (in particular, when this procedure
is implied by the need of measuring while the detector is moving relative
to the primary circuit, if the intricacies of a circuit continuously
deforming are to be avoided). Thus, the scientist can perceive (measure)
currents and charges using the original defining method if at rest
relative to the source and the effects (as encrypted in forces, fields
and potentials) of such events if the observer is at rest relative
to the receiver. Currents and charges in the source are only inferred
by the observer at rest with the receiver while forces, fields and
potentials are inferred by the observer at rest with respect to the
source.

We have shown that Electromagnetism can be formulated in terms of
fields associated to sources as well as fields associated to receivers,
this symmetry is broken in a construction that focuses exclusively
in S-fields rather than R-fields. Remark \ref{Rem:SR-fields} exposes
the relation between S- and R- fields, while the relation between
source and receiver descriptions is given for example by the pair
of equations \eqref{eq:kernel} and \eqref{eq:motion}. The restoration
of this symmetry explains how the action-reaction law of Newton's
mechanics is identically broken (Subsection \ref{subsec:Source/receiver-symmetry})
in the standard construction of Electromagnetism.

Most interestingly, the present approach based on eqs. \eqref{eq:kernel},
\eqref{eq:properJ} and \eqref{eq:motion}, is consistent with normal
Electrodynamics and explains two fundamental concerns of the original
theory, namely that electromagnetic waves propagate with the same
electromagnetic parameter $C$ regardless of the state of relative
motion between source and detector, and that the electromagnetic Doppler
effect is acted by a Lorentz boost of parameter $u$, in agreement
with the accepted description. These results are obtained within the
original framework of the theory, in particular preserving the Euclidean
character of the auxiliary space-time, $\mathbb{R}^{3+1}=\mathbb{R}^{3}\times\mathbb{R}^{1}$,
which fulfils the conditions imposed by spatial relations or relational
space. In terms of interpretations, there is no need to regard the
universal constant ${\displaystyle C=(\mu_{0}\epsilon_{0})^{-\frac{1}{2}}}$
as a velocity, nor to have something travelling between source and
detector when considering electromagnetic interactions.

The complete set of equations of electromagnetism (Maxwell's equations,
continuity equation and Lorentz' force) arise in the present form
as the result of postulating Lorenz' delayed-action-at-a-distance
\ref{eq:Lorenz} and Lorentz' action integral, to be used in the principle
of least action \ref{eq:action}. Lorenz' postulate has empirical
basis while Lorentz' action is a (mathematical) organisation principle
that has been considered fundamental by several authors as for example
Poincaré \citep{poin13}. It is interesting to notice that before
the irruption of the ``second physicist'' \citep{jung17}, i.e.,
the theoretical physicist, theory in physics had a meaning close to
``mathematically organised empirical observations''. This is the
spirit of Maxwell's work but it is as well the spirit in Newton, Ampere,
Gauss and many others in the earlier times of physics. This epistemic
position was heavily attacked by proponents of the ether such as Heaviside
\citep{heaviside2011electrical}, Hertz \citep{hert93} and particularly
Clausius \citep{clau69} who directly attacked Gauss' conception in
the works by Riemann \citep{riem67}, Betti \citep{bett67} and Neumann
\citep{neum68}.

The present formulation addresses an issue recognised
by \citet[p. 228]{maxwell1990scientific}: 
\begin{quote}
... According to a theory of electricity which is making great progress
in Germany, two electrical particles act on one another directly at
a distance, but with a force which, according to Weber, depends on
their relative velocity, and according to a theory hinted at by Gauss,
and developed by Riemann, Lorenz, and Neumann, acts not instantaneously,
but after a time depending on the distance. The power with which this
theory, in the hands of these eminent men, explains every kind of
electrical phenomena must be studied in order to be appreciated {[}...{]}
\end{quote}
And comparing with his preferred theory that ``attributes electric
action to tensions and pressures in an all-pervading medium'' he
writes:
\begin{quote}
That theories apparently so fundamentally opposed should have so large
a field of truth common to both is a fact the philosophical importance
of which we cannot fully appreciate till we have reached a scientific
altitude from which the true relation between hypotheses so different
can be seen. 
\end{quote}
About one and a half century after Maxwell's conference we can discuss
his philosophical inquire. Both theories are in perfect mathematical
correspondence as they are with currently accepted electromagnetism
but they differ in the abduction and interpretation as well as in
the use of auxiliary concepts. Current electromagnetism relies heavily
on the inferred idea of space-time and a mechanical analogy of the
interaction. This approach is effective but leads us to embrace a
new form of space-time, a necessary belief that not all of us are
willing to admit. More precisely, current electromagnetism constructs
first the space (relating it to Lorentz transformations) and only
next spatial relations. In our view this order leads to logical inconsistencies
\citep{Solari22-phenomenologico} at the time of construction, despite
the success achieved in terms of experimental comparisons. The present
approach solves the problem by disposing of the subjective (auxiliary)
space resting directly on spatio-temporal relations in such a way
that rather than resting on just one transformation, a shared attribute
of previous approaches, we find a harmonious coexistence of Galilean
and Lorentzian transformations, a sort of reconciliatory mid-point.
To achieve this views we had to accept first that space and time are
not an a priori of knowledge as Kant thought \citep{kant87} but rather
a construction of the child as Piaget experimentally found \citep{piag99}.
Moreover, all these apparently conflicting approaches are needed for
science to progress.

This view only uses (subjective) space and time as an auxiliary element
when and if needed. The symmetry associated to the arbitrary decision
of using a reference system is the one expressed by Galilean transformations
\citep{sola18b}. In this context, inertial systems as auxiliary reference
systems are constructed on the basis of the idea of free-bodies \citep{newt87,thomson1884,sola22}.
This structure is underlying the work but not explicitly used as we
have preferred to avoid reference frames.

The Lorentz transformations correspond in this construction to an
endomorphism of relational space-time and are relevant only to the
propagation of electromagnetic action. They are associated to the
(unexpected) symmetry concerning perception and inferred charges and
currents. While the full Poincaré-Lorentz group relates different
but equivalent perceptions of the electromagnetic action, the particular
set of Lorentz transformations relate the perception from a detector
at rest with respect to the source to the perception of a detector
in motion relative to the source with (invariant) velocity $u$. Such
relation can be obtained only for $|u|<C$. We emphasise that the
Poincaré-Lorentz group coexists with the Galilean symmetry of the
description, although not all the equations, particularly not all
differential equations, are transparent as expressions showing the
symmetry. The integral presentation is in this respect more revealing.

We finally stress that symmetries as requirements of reason pre-exist
physics and equations. They enter physics as a demand of reason in
our quest to construct the cosmos, this is, to put in harmony our
perceptions of the real-sensible.

The ether was the immediate consequence of attempts to understand
electromagnetism by analogy to mechanical phenomena. Special relativity
introduced an analogy of the forms, the Principle of relativity, without
an understanding of the fundaments of the principle. It soon became
evident that if analogies with mechanics would be preserved, the metric
of space-time had to be changed. Yet, hiding the hypothesis the statement
reads: electromagnetism imposes us to adopt a different metric of
space-time than the Cartesian one used in its construction. Next,
to accept this unmatching between the construction moment and the
explanatory moment of science requires the exclusion of the first,
leaving us with a science without understanding, supported only upon
its predictive success, a technology of prediction, since success
is the quality measure of any technology. Philosophers like Popper
and Reichenbach considered their task to support the theories of scientists
like Einstein. Consequently, they dropped all critical examination
of matters, finally endorsing a program that was put forward by 1870,
``physics must henceforth pursue the sole aim of writing down for
each series of phenomena ... equations from which the course of the
phenomena can be quantitatively determined; so that the sole task
of physics consisted in using trial and error to find the simplest
equations''. Notice that even ``trial and error'', the method favoured
by Popper, was already indicated. Such program is the instrumentalist
program of science, aiming at dominating nature, a perfect mate of
considering the Earth an infinite source of resources for the development
of the capitalist society. We argue then that it has been forced upon
us by social decisions that made nearly impossible the survival of
the critical motion of reason. Conversely, by restating critical reasoning,
we have been able to construct an electromagnetism that is more consilient
than the received wisdom, it does not need to reform space-time and
consequently makes no call for the abandonment of the construction
moment. Reason can organise the chaos that reaches our senses, harmony
is still an enticing possibility. If the child develops abstraction
to understand the possible instead of being forced to accept the given,
we need to put abstraction to work. We have been told that there is
just one possible science, the given science, the science of capitalism.
We have proved by presenting a counter example that the statement
is wrong. Science is not only what scientists do (the given) but what
humans can do as well, critical and ethical science, a science conscious
of its ignorance. We close with ancient words by Chuang Tzu:
\begin{quote}
Now you have come out beyond your banks and borders and have seen
the great sea -- so you realize your own pettiness. From now on it
will be possible to talk to you about the Great Principle. \citep[Autumn Floods]{ChuangTzu68}
\end{quote}

\section*{Acknowledgements}

We thank Alejandro Romero Fernández, Olimpia Lombardi and Federico
Holik for valuable discussions.

\section*{Declaration of Interest}

The authors declare that there exists no actual or potential conflict of interest including any financial,
personal or other relationships with other people or organizations within three years of beginning the
submitted work that could inappropriately influence, or be perceived to influence, their work.

%\bibliographystyle{plainnat}
%\bibliography{nuevasreferencias,referencias}

\appendix

\section{Some Proofs}

\subsection{Proof of Lemma \ref{lem:wave}\label{Proof-lem-wave}}
\begin{proof}
We perform the calculation in detail only for $A$, since the other
one is similar. We use the shorthand $r=|x-y|$.
\begin{eqnarray*}
\nabla_{x}A_{i} & = & \frac{\mu_{0}}{4\pi}\int d^{3}y\,\left(j_{i}\nabla_{x}\frac{1}{r}-\frac{\frac{\partial}{\partial t}j_{i}\nabla_{x}\frac{r}{C}}{r}\right)\\
\Delta A_{i}
	& = & \nabla_{x}\cdot\nabla_{x}A_{i}\\
	& = & \frac{\mu_{0}}{4\pi}\int d^{3}y\,\left(j_{i}\Delta\frac{1}{r}-2\left(\nabla_{x}\frac{1}{r}\right)\cdot\left(\frac{\partial}{\partial t}j_{i}\nabla_{x}\frac{r}{C}\right)
	\right)- \\ & & -\frac{\mu_{0}}{4\pi}\int d^{3}y\,  \left(
	\frac{\frac{\partial}{\partial t}j_{i}\Delta\frac{r}{C}}{r}+\frac{\frac{\partial^{2}}{\partial t^{2}}j_{i}}{r}|\nabla_{x}\frac{r}{C}|^{2}\right)
\end{eqnarray*}
Moreover, standard vector calculus identities give
\begin{eqnarray*}
\frac{\partial}{\partial t}j_{i}\left(2\nabla\frac{1}{r}\cdot\nabla\frac{r}{C}+\frac{\Delta\frac{r}{C}}{r}\right)
	& = & 0\\ |\nabla\frac{r}{C}|^{2} & = & \frac{1}{C^{2}}
\end{eqnarray*}
and therefore
\begin{eqnarray*}
\Delta A_{i}(x,t) & = & \frac{\mu_{0}}{4\pi}\int d^{3}y\,j_{i}(y,t-\frac{r}{C})\Delta\left(\frac{1}{r}\right)+\left(\frac{1}{C^{2}}\right)\frac{\mu_{0}}{4\pi}\int d^{3}y\,\frac{\partial^{2}}{\partial t^{2}}\frac{j_{i}(y,t-\frac{r}{C})}{r}
\end{eqnarray*}
The time derivative in the last term can be extracted outside the
integral, thus yielding,
\begin{eqnarray*}
\Box A_{i}(x,t) & = & \Delta A_{i}(x,t)-\left(\frac{1}{C^{2}}\right)\frac{\mu_{0}}{4\pi}\int d^{3}y\,\frac{\partial^{2}}{\partial t^{2}}\frac{j_{i}(y,t-\frac{r}{C})}{r}\\
 & = & \Delta A_{i}(x,t)-\left(\frac{1}{C^{2}}\right)\frac{\partial^{2}}{\partial t^{2}}A_{i}(x,t)\\
 & = & \frac{\mu_{0}}{4\pi}\int d^{3}y\,j_{i}(y,t-\frac{|x-y|}{C})\Delta\left(\frac{1}{r}\right)\\
 & = & -\mu_{0}j_{i}(x,t)
\end{eqnarray*}
\end{proof}

\subsection{Proof of Theorem \ref{thm:wave}\label{Proof-thm-wave}}
\begin{proof}
The result follows from the computation of the extremal action under
the constraints
\begin{eqnarray*}
(V-\mathsf{V})\chi & = & 0\\
(A-\mathsf{A})\chi & = & 0
\end{eqnarray*}
Multiplying the constraints by the Lagrange multipliers $\lambda$
and $\kappa$ (the latter a vector), while we use the shorthand notations
$B=\nabla\times A$ and $E=\left(-\frac{\partial A}{\partial t}-\nabla V\right)$,
we need to variate the constrained electromagnetic action
\[
{\cal A}=\frac{1}{2}\int dt\left(\int\left(\frac{1}{\mu_{0}}|B|^{2}-\epsilon_{0}|E|^{2}-\kappa\cdot(A-\mathsf{A})\chi+\lambda(V-\mathsf{V})\chi\right)\,d^{3}x\right).
\]
Varying the integrand we obtain
\begin{eqnarray*}
\delta\mathcal{A} & = & \int\,dt\left(\int\left(\frac{1}{\mu_{0}}\left(\nabla\times A\right)\cdot\left(\nabla\times\delta A\right)-\epsilon_{0}\left(\nabla V\cdot\nabla\delta V+\frac{\partial A}{\partial t}\frac{\partial\delta A}{\partial t}\right)\right.\right.\\
 &  & \left.\left.-\epsilon_{0}\left(\nabla V\cdot\frac{\partial\delta A}{\partial t}+\nabla\delta V\cdot\frac{\partial A}{\partial t}\right)-\chi\kappa\,\delta A+\chi\lambda\,\delta V\right)\,d^{3}x\right)
\end{eqnarray*}
Partial integrations in time and standard vector calculus give the
following identities: 
\begin{eqnarray*}
\int dt\,\frac{\partial A}{\partial t}\frac{\partial\delta A}{\partial t} & = & \left[\delta A\cdot\frac{\partial A}{\partial t}\right]-\int dt\,\delta A\cdot\frac{\partial^{2}A}{\partial t^{2}}\\
\int dt\,\nabla V\cdot\frac{\partial\delta A}{\partial t} & = & \left[\delta A\cdot\nabla V\right]-\int dt\,\delta A\cdot\nabla\frac{\partial V}{\partial t}\\
\left(\nabla\times A\right)\cdot\left(\nabla\times\delta A\right) & = & \nabla\times\left(\nabla\times A\right)\cdot\delta A-\left[\nabla\cdot\left(\left(\nabla\times A\right)\times\delta A\right)\right]\\
\left(\nabla V+\frac{\partial A}{\partial t}\right)\cdot\nabla\delta V & = & \left[\nabla\cdot\left(\left(\nabla V+\frac{\partial A}{\partial t}\right)\delta V\right)\right]-\delta V\nabla\cdot\left(\nabla V+\frac{\partial A}{\partial t}\right)
\end{eqnarray*}
The terms in square brackets vanish in the variation either for the
vanishing variation at endpoints or because of Gauss theorem applied
to functions decaying fast enough at infinity. Hence, 
\begin{eqnarray*}
\delta\mathcal{A} & = & \int dt\,\int d^{3}x\,\left(\frac{1}{\mu_{0}}\left(\nabla\times A\right)\cdot\left(\nabla\times\delta A\right)-\epsilon_{0}\left(-\delta V\nabla\cdot\left(\nabla V+\frac{\partial A}{\partial t}\right)\right.\right.\\
 &  & \left.\left.-\delta A\cdot\frac{\partial^{2}A}{\partial t^{2}}-\delta A\cdot\nabla\frac{\partial V}{\partial t}\right)-\chi\kappa\cdot\delta A+\chi\lambda\delta V\right).
\end{eqnarray*}
Being $\delta A$ and $\delta V$ independent, we obtain 
\begin{eqnarray*}
\frac{1}{\mu_{0}}\nabla\times\left(\nabla\times A\right)+\epsilon_{0}\left(\frac{\partial^{2}A}{\partial t^{2}}+\nabla\frac{\partial V}{\partial t}\right) & = & \chi\kappa\\
-\epsilon_{0}\nabla\cdot\left(\nabla V+\frac{\partial A}{\partial t}\right) & = & \chi\lambda
\end{eqnarray*}
or equivalently 
\begin{eqnarray*}
\frac{1}{\mu_{0}}\nabla\times B-\epsilon_{0}\frac{\partial E}{\partial t} & = & \chi\kappa\\
\epsilon_{0}\nabla\cdot E & = & \chi\lambda
\end{eqnarray*}
which allows us to identify $j=\chi\kappa$ (the density of current
inside the material responsible for $\mathsf{A}$) and ${\displaystyle \rho=}\chi\lambda$
(the density of charge responsible for $\mathsf{V}$), thus proving
the first result. Finally, the continuity equation follows from0
\[
0=\nabla\cdot\left(\frac{1}{\mu_{0}}\nabla\times B-\epsilon_{0}\frac{\partial E}{\partial t}\right)+\frac{\partial}{\partial t}\left(\epsilon_{0}\nabla\cdot E\right)=\nabla\cdot j+\frac{\partial\rho}{\partial t}.
\]

Note also that taking curl on the first equation we verify that $B$
satisfies a wave equation. Inserting $\nabla\times E$ in the time-derivative
of the first equation and adding the gradient of the second equation,
we obtain a wave equation for $E$.

Further, $\frac{1}{\mu_{0}}\nabla\cdot A+\epsilon_{0}\frac{\partial V}{\partial t}=0$
implies both that ${\displaystyle \nabla\frac{\partial V}{\partial t}=-C^{2}\nabla\left(\nabla\cdot A\right)}$
and ${\displaystyle \nabla\cdot\frac{\partial A}{\partial t}=-\frac{1}{C^{2}}\frac{\partial^{2}V}{\partial t^{2}}}$.
Substituting each relation in the corresponding equation, we obtain
eq(\ref{eq:diff-wave}), thus proving the Corollary.
\end{proof}

\subsection{Proof of Theorem \ref{Thm-NoLor}\label{subsec:Proof-of-NoLor}}
\begin{proof}
The rotational invariance is immediate, since for any rotation matrix
$R,$the change of coordinates $x^{\prime}=Rx$ (along with the corresponding
change for $y$), keeps the distance $|R(x-y)|=|x-y|$ invariant.
Hence, for the kernel in eq.(\ref{eq:kernel}) and any electromagnetic
kernel depending on $|x-y|$ the action integral is invariant under
rotations. Let $u$ be the velocity associated to a Lorentz transformation,
which is constant by hypothesis. The proposed variation reads 
\begin{eqnarray*}
\hat{\delta}\mathcal{A} & = & \int_{-\infty}^{t}ds\int d^{3}x\,\hat{\delta}\left[A^{1}(x,s)j_{2}(x,s)-V^{1}(x,s)\rho_{2}(x,s)\right]\\
 & = & \int_{-\infty}^{t}ds\int d^{3}x\,\left(Cs\,\delta u\cdot\nabla_{x}+(\frac{x}{C}\cdot\delta u)\partial_{s}\right)\left[A^{1}j_{2}-V^{1}\rho_{2}\right](x,s)
\end{eqnarray*}
where ${\displaystyle \mathcal{I}(x,s)=Cs\,\delta u\cdot\nabla_{x}+(x\cdot\delta u)\partial_{Cs}}$
is the Lorentz generator. By Gauss Theorem the following integral
vanishes for any function $F$ inheriting the behaviour of $A,j$
at infinity:
\begin{eqnarray*}
\int_{K}d^{3}x\,\delta u\cdot\nabla F(x) & = & \int_{\partial K}F(x)\delta u\cdot dS=0
\end{eqnarray*}
Finally, 
\begin{eqnarray*}
\int_{-\infty}^{t}ds\,\frac{\partial}{\partial s}\int d^{3}x\,(x\cdot\delta u)G(x,s) & = & \mathcal{F}(t)-\mathcal{F}(t_{0})
\end{eqnarray*}
for some function $\mathcal{F}$ depending only of $t$. However by
the nature of the variational process, $\mathcal{F}$ does not contribute
to the variation.
\end{proof}

\section{The Lorentz transformation\label{sec apendice LT}}

The infinitesimal generator of the Lorentz transformation in eq.(\ref{eq:TL-coord})
reads 
\[
{\cal I}_{j}=\left(\begin{array}{cc}
\mathbf{0} & \frac{v}{C}\\
(\frac{v}{C})^{T} & 0
\end{array}\right).
\]
The Lorentz transformation for finite $v$ is obtained by exponentiation
\citep{gilm74}, yielding the $4\times4$ matrix expression $TL(v)$
for the transformation elements, where
\[
TL(v)=\left(\begin{array}{cc}
W & X\\
X^{\dagger} & Y
\end{array}\right)
\]
is formed by the $3$-vector $X=\sinh\left(\left|\frac{v}{C}\right|\right){\displaystyle \frac{v}{|v|}}$,
the scalar $Y=\sqrt{1+|X|^{2}}={\displaystyle \cosh\left(\left|\frac{v}{C}\right|\right)}$
and the $3\times3$ matrix $W=Id+\left(\cosh\left(\left|\frac{v}{C}\right|\right)-1\right){\displaystyle \frac{vv^{\dagger}}{|v|^{2}}}$,
where $v\in\mathbb{R}^{3}$ is a parameter classifying the different
transformations. A better known expression for the Lorentz transformation
arises from the change of variables $u={\displaystyle C\hat{v}\,\tanh\left|\frac{v}{C}\right|}$
\citep{gilm74}. In such terms,
\[
L(u)=\left(\begin{array}{cc}
Id+(\gamma-1)\hat{u}\hat{u}^{\dagger} & {\displaystyle \gamma\frac{u}{C}}\\
{\displaystyle \gamma\frac{u^{\dagger}}{C}} & \gamma
\end{array}\right);\quad L\left(u(v)\right)\equiv TL(v),
\]
where $\gamma\left(u(v)\right)={\displaystyle \frac{1}{\sqrt{1-\left|\frac{u(v)}{C}\right|^{2}}}}=\cosh\left(\left|\frac{v}{C}\right|\right)$.
Note that for ${\displaystyle {\displaystyle \left|\frac{v}{C}\right|\ll1}}$,
we have that ${\displaystyle \left|\frac{u}{C}\right|=\left|\frac{v}{C}\right|+O\left(|\frac{v}{C}|^{3}\right)}$.

\end{document}